\documentclass[12pt]{amsart}

\usepackage{amsmath,amsfonts,amsbsy,amsthm}

\usepackage{dsfont}
\usepackage{mathtools}
\usepackage{mathrsfs}

\usepackage{enumerate}
\usepackage{enumitem}
\setlist{leftmargin=*}

\numberwithin{equation}{section}

\makeatletter
\newtheoremstyle{corsivo}
   {\medskipamount}{\medskipamount}%
   {\itshape}{}%
   {\bfseries}{}%
   { }
   {\thmname{#1}\thmnumber{\@ifnotempty{#1}{ }\@upn{#2}}%
    \thmnote{ {\bfseries\boldmath(#3)}}.}%
\makeatother

\theoremstyle{corsivo}
\newtheorem{theorem}{Theorem}[section]
\newtheorem{lemma}[theorem]{Lemma}
\newtheorem{corollary}[theorem]{Corollary}
\newtheorem{proposition}[theorem]{Proposition}

\makeatletter
\newtheoremstyle{dritto}
   {\medskipamount}{\medskipamount}%
   {\rmfamily}{}%
   {\bfseries}{}%
   { }
   {\thmname{#1}\thmnumber{\@ifnotempty{#1}{ }\@upn{#2}}%
    \thmnote{ {\bfseries\boldmath(#3)}}.}%
\makeatother

\theoremstyle{dritto}
\newtheorem{definition}[theorem]{Definition}
\newtheorem{remark}[theorem]{Remark}

\newtheorem{assumption}[theorem]{Assumption}

\newcommand{\sub}[1]{_{\mathrm{#1}}}
\newcommand{\su}[1]{^{\mathrm{#1}}}

\newcommand{\eps}{\varepsilon}
\newcommand{\epsi}{\varepsilon}

\newcommand{\Id}{\mathds{1}}  

\newcommand{\eu}{\mathrm{e}}
\newcommand{\iu}{\mathrm{i}}   
\newcommand{\di}{\mathrm{d}}

\newcommand{\N}{\mathbb{N}}
\newcommand{\Z}{\mathbb{Z}}
\newcommand{\R}{\mathbb{R}}
\newcommand{\C}{\mathbb{C}}

\newcommand{\Do}{\mathcal{D}}

\newcommand{\Hi}{\mathcal{H}}

\newcommand{\Hf}{\mathcal{H}\sub{f}}
\newcommand{\Df}{\mathcal{D}\sub{f}}
\newcommand{\U}{\mathcal{U}}
\newcommand{\X}{\mathcal{X}}
\newcommand{\PO}{\mathcal{P}}

\newcommand{\UZ}{\U\sub{BFZ}}   

\newcommand{\norm}[1]{\left\| #1 \right\|}

\newcommand{\set}[1]{ \left\{  #1 \right\}} 

\DeclareMathOperator{\Tr}{Tr}         

\DeclareMathOperator{\Ran}{Ran}

\newcommand{\eg}{{\sl e.\,g.\ }} 

\newcommand{\abs}[1]{\left\lvert#1\right\rvert}

\newcommand{\FC}{{\mathcal{C}}}

\newcommand{\SC}{\mathcal{S}}

\newcommand{\B}{\mathcal{B}}

\newcommand{\LI}{\mathscr{L}}
\usepackage{comment}

\newcommand{\virg}[1]{``#1''}

\renewcommand{\(}{\left(}
\renewcommand{\)}{\right)}

\newcommand{\E}{{\mathrm{e}}}
\newcommand{\I}{\mathrm{i}}
\newcommand{\Or}{{\mathcal{O}}}

\let\oldfootnote\footnote
\renewcommand{\footnote}[1]{\oldfootnote{\  #1}}

\title[Purely linear response to space-adiabatic perturbations]{Purely linear response of the quantum Hall current to space-adiabatic perturbations}
\author[G.~Marcelli, D.~Monaco]{Giovanna Marcelli \and Domenico Monaco}
\date{\today, arXiv version 1}

\usepackage{hyperref}

\begin{document}

\begin{abstract}
Using recently developed tools from space-adiabatic perturbation theory, in particular the construction of a non-equilibrium almost stationary state, we give a new proof that the Kubo formula for the Hall conductivity remains valid beyond the linear response regime. In particular, we prove that, in quantum Hall systems and Chern insulators, the transverse response current is quantized up to any order in the strength of the inducing electric field. The latter is introduced as a perturbation to a periodic, spectrally gapped equilibrium Hamiltonian by means of a linear potential; existing proofs of the exactness of Kubo formula rely instead on a time-dependent magnetic potential. The result applies to both continuum and discrete crystalline systems modelling the quantum (anomalous) Hall effect.
\end{abstract}

\maketitle

\tableofcontents

\goodbreak


\section{Introduction and main results}

The mathematical understanding of transport properties of quantum system is a fundamental question in the mathematical physics of condensed matter, and still to date provides a stimulating challenge. The interest in this line of research has increased further after the discovery of ``exotic'' transport phenomena of topological origin, most notably the \emph{quantum Hall effect}, where a 2-dimensional electron gas subject to a perpendicular magnetic field displays a transverse current in response to an inducing in-plane electring field of strength $\eps$: at zero temperature, the conductivity for this transverse current can be computed by Kubo's formula at least in the linear response regime \cite{Kubo57}, and appears experimentally to be quantized (in appropriate physical units) to astounding precision \cite{vonKlitzing80}. More recently, a similar topological transport has been observed in \emph{Chern insulators}, where time-reversal symmetry is broken by a different mechanism than an external magnetic field \cite{Haldane88, Bestwick_et_al, cinesi1, cinesi2}: this phenomenon is then called \emph{quantum anomalous Hall effect}.

The formula by Kubo expresses the transverse Hall current~$j$ as
\[ j = \eps \, \sigma\sub{Hall} + \Or(\eps), \]
where $\sigma\sub{Hall} \in (e^2/h)\,\Z$ is expressed in terms of equilibrium quantities (see below), involving in particular the Fermi projection $\Pi_0$ onto occupied energy levels. Within the one-particle picture, the quantization of the Hall conductivity has been by now understood mathematically by means of its connection with the \emph{Chern number} of the Fermi projection from differential geometry, and with its noncommutative generalization, the \emph{Chern marker} (see \cite{Graf07} for a comprehensive review). Recent mathematical efforts have managed to extend results in this direction also to the setting of electrons interacting on a lattice \cite{HastingsMichalakis15, GiulianiMastropietroPorta17, BachmannDeRoeckFraas17, MonacoTeufel17}. We refer the reader to the recent review \cite{HenheikTeufel21} for further comments on the (mathematical) literature on the Kubo formula.

The topological nature of the Hall conductivity $\sigma\sub{Hall}$ is believed to be responsible for its stability and robustness, making it \emph{universal}, that is, independent of specific features of the model. Furthermore, its geometric origin is responsible also for the fact that the validity of the Kubo formula extends well beyond linear response: indeed, the conductivity associated to the transverse current of Hall systems is known to be equal to $\sigma\sub{Hall}$ up to \emph{arbitrarily high orders} in the strength $\eps$ of the perturbing electric field, that is,
\begin{equation} \label{KuboToAllOrders}
j = \eps \, \sigma\sub{Hall} + \Or(\eps^\infty).
\end{equation}
The existing literature on this property was initiated by the heuristic \emph{magnetic flux insertion argument} proposed by Laughlin in a cylindrical geometry \cite{Laughlin81}, which was later elaborated in a rigorous way for many-body electron gases in the continuum \cite{KleinSeiler} or discrete \cite{Bachmannetal} setting. These proofs focus on a related quantity, namely the Hall \emph{conductance}, defined as the (linear) response of the current intensity to the voltage drop: in two dimensions, this quantity agrees with the Hall \emph{conductivity} $\sigma\sub{Hall}$ defined above, see \cite{AvronSeilerSimon}. In the magnetic-flux-insertion argument, the inducing electric field is modelled by a slowly-varying time-dependent magnetic potential: this allows to follow time-adiabatically the insertion of this magnetic flux in the ground state. Klein and Seiler~\cite{KleinSeiler} then make use of the geometric interpretation of $\sigma\sub{Hall}$ to conclude the validity of (the Hall-conductance analogue of) \eqref{KuboToAllOrders}, at least up to averaging over time and over the inserted magnetic flux. Instead, Bachmann \textsl{et al.}~\cite{Bachmannetal} obtain an analogous statement (in the context of lattice spin systems with local interactions and observables) avoiding magnetic-flux averaging and the geometric argument, at the expense of exploiting the integrality of a certain Fredholm index related to the Hall conductance. Both approaches rely on the assumption that this magnetic flux insertion does not close the gap of the unperturbed Hamiltonian.

\medskip

In this paper, we manage to prove the above-mentioned remarkable property of the Hall conductivity (Theorem~\ref{thm:main result}) avoiding the magnetic flux insertion altogether. To model the equilibrium system, we employ a spectrally gapped (insulating) one-particle Hamiltonian~$H_0$; this could be a discrete, tight-binding Hamiltonian, or a continuum (magnetic) Schr\"odinger-type operator. Contrary to the above-mentioned references, the external electric field will then be introduced by the addition of a linear potential to the equilibrium Hamiltonian, closer to how experimental setups for the quantum Hall effect were originally performed. Our argument relies on two main tools:
\begin{enumerate}
 \item by treating the linear electric potential as a \emph{space-adiabatic perturbation}, we are able to construct a \emph{non-equilibrium almost stationary state} (NEASS), in the sense of \cite{Teufel20, MaPaTe}, which in the adiabatic regime well approximates the physical state of the system once the dynamical switching drives the Fermi projection out of equilibrium \cite{MarcelliTeufel}; 
 \item the connection of the conductivity associated to the current flowing in the NEASS with its topological value $\sigma\sub{Hall}$ is realized in our proof by a \emph{Chern--Simons-like formula} (Proposition~\ref{prop:CS formula}), similar to that used in \cite{KleinSeiler}.
\end{enumerate}

By definition, the NEASS is unitarily conjugated to the equilibrium Fermi projection (see property \ref{item:SA1} below): this structure is reminiscent of the ``magnetic gauge transformed projection'' of \cite{KleinSeiler}, as well as of the ``dressed ground state'' of \cite{Bachmannetal}, with the main difference that the unitary conjugation is defined here by employing \emph{space-adiabatic} rather than time-adiabatic perturbation theory. The NEASS was constructed in \cite{MaPaTe} up to first order in $\eps$ in the same context that we will employ; using arguments from \cite{Teufel20}, we extend this construction to arbitrarily high orders in $\eps$ (Theorem~\ref{thm:constr NEASS}), a result which is interesting in its own right.  

Since we deal with extended systems, a prominent role is played by the \emph{trace per unit volume} $\tau(\cdot)$, which is used to compute expectation values of extensive observables in (non-)equilibrium states: for example, the charge current which flows in the NEASS $\Pi^\eps$ equals
\[ \tau(\iu [H_0, X] \, \Pi^\eps), \]
and the quantized value of the Hall conductivity can be expressed (in appropriate physical units for which $e = \hbar = 1$) as
\[ \sigma\sub{Hall} = \iu \, \tau\left( \Pi_0 \big[ [\Pi_0, X], [\Pi_0, Y] \big] \right) \in \frac{1}{2\pi} \, \Z, \]
if the electric field is applied, say, in direction $y$, and the transverse current is measured along direction $x$. To ensure the well-posedness of all traces per unit volume that need to be considered, especially in continuum systems, we restrict ourselves to the setting of crystalline systems and \emph{periodic} operators (that is, operators which commute with translations by crystalline shifts in a Bravais lattice), and introduce certain operator algebras of such operators (see Section~\ref{sec: op alg}). The heart of the proof is however of ``algebraic'' nature, and therefore we believe that it may be generalized to apply to settings which also include ergodic disorder (at least under a spectral gap assumption), in which the relevant operators satisfy only a covariance property when shifted by lattice translations (see e.g.~\cite{BoucletGerminetKleinSchenker05} for a framework of this type).

Finally, let us comment on the applicability of our result to \emph{spin transport}. The discovery of topological insulators in the early 2000's stimulated the study of topological transport of spin, for example in the quantum spin Hall effect. When spin is a conserved quantity, $[H_0, S^z] = 0$, a spin current operator can be defined as $J^z = \iu [H_0, X]\,S^z$, and the response of this current to an external electric field can be also studied. As it is easily realized, this setting essentially amounts to two ``copies'' of a quantum Hall system (one corresponding to charge carriers with ``up'' spin, and one to those with ``down'' spin), and our result applies to this spin-filtered charge transport as well, leading to the quantization of the Hall conductivity in each spin channel separately. A much richer and mathematically more challenging situation would be to consider systems in which spin is \emph{not} conserved, for example due to the presence of Rashba spin-orbit coupling in the model. While formul\ae\ for the (appropriate generalization of) spin conductivity have been already investigated analytically \cite{MarcelliPanatiTauber19, MaPaTe} and numerically \cite{MonacoUlcakar20} within linear response, the existence of possible power-law correction to these formul\ae\ remains to be studied. We postpone this investigation to future work.

\medskip

The paper is structured as follows. Section~\ref{sec:model} details the class of models to which our result applies. Section~\ref{sec:NEASSAO} provides the construction of the NEASS to all orders in $\eps$, generalizing the results of \cite{MaPaTe} beyond the linear regime. Section~\ref{sec:main result} finally contains the proof of our main result stating the validity of the Kubo formula for the Hall conductivity to arbitrarily high orders of $\eps$. The Appendices contain, for the readers' convenience, some properties of the trace per unit volume and of the inverse Liouvillian of a gapped Hamiltonian, which are used throughout the paper.

\medskip

\noindent{\bf Acknowledgements.}
The authors are grateful to G.~Panati, M.~Porta and S.~Teufel for their careful reading and valuable comments on a first draft of the paper.
This work was supported by the National Group of Mathematical Physics (GNFM--INdAM) within the project \virg{Progetto Giovani GNFM 2020}. G.\,M. gratefully acknowledges the financial support from the European Research Council (ERC), under the European Union's Horizon 2020 research and innovation programme (ERC Starting Grant MaMBoQ, no. 802901).

\section{Model and mathematical framework} \label{sec:model}

\subsection{Crystalline structures and periodic operators}

The quantum systems we will be analyzing have a \emph{crystalline structure}, meaning that their configuration space~$\X$ is invariant under translations by vectors in a Bravais lattice $\Gamma$. We will address both continuum and discrete models on the same footing: in $d$-dimensions, by a continuum configuration space we mean $\X = \R^d$, while a discrete configuration space is a discrete set of points. In both cases, it can be assumed that the Bravais lattice $\Gamma$ is spanned over the integers by a basis $\set{a_1, \ldots, a_d} \subset \R^d$. 

The Hilbert space for a quantum particle with $N$ internal degrees of freedom (say, spin) will be
\[
\Hi := L^2(\X) \otimes \C^N \simeq L^2(\X ,\C^N ).
\]
A prominent feature of this Hilbert space is the possibility to define (self-adjoint) \emph{position operators}: 
\[
 (X_j \psi)(x ) := x_j \, \psi(x ), \quad 1\leq j \leq d.
\]
The above definition of course makes sense only on a suitable (maximal) domain $\Do(X_j) \subset \Hi$.

The crystalline structure of the configuration space is lifted to a symmetry of the one-particle Hilbert space, namely we assume that there is a unitary representation $T \colon \Gamma \to \mathcal{U}(\Hi)$, $\gamma \mapsto T_\gamma$, by \emph{translation operators}. Let us note that, in presence of uniform magnetic fields, these operators could be magnetic translations \cite{Zak64}, assuming a commensurability condition on the magnetic flux per unit cell and the quantum of magnetic flux. These considerations are relevant for quantum Hall systems, which are included in our framework under the above-mentioned commensurability hypothesis. 

An operator $A$ on $\Hi$ is called \emph{periodic} if $[A, T_\gamma] = 0$ for all $\gamma \in \Gamma$. As is well-known, the analysis of periodic operators is simplified by the use of the \emph{(magnetic)  
Bloch--Floquet--Zak representation} (see \eg \cite{FreundTeufel16} and references therein), which introduces the crystal momentum $k \in \R^d$ as a good quantum number. The {\emph{(magnetic) Bloch--Floquet--Zak transform}} is initially defined on compactly supported functions $\psi\in C_0 (\X,\C^N)\subset L^2(\X,\C^N)$ as
\begin{equation}
\label{eqn:defn UZ}
(\UZ \psi)(k,y) :=\eu^{-\iu k \cdot y} \sum_{\gamma \in \Gamma} \eu^{\iu k \cdot \gamma} (T_\gamma \psi)(y)\qquad \text{for all } k\in\R^d,\, y\in \X.
\end{equation}
For fixed $k\in \R^d$, the function $(\UZ \psi)(k,\cdot)$ is periodic with respect to the translations operators, hence it defines an element in the so-called fiber Hilbert space 
\[ 
\Hf := \left\{ \phi \in L^2\sub{loc}(\X,\C^N)\,|\,  T_\gamma \phi=\phi \mbox{ for all $\gamma\in\Gamma$} \right\} 
\]
which is equipped with the scalar product induced by the norm
\[ \|\phi\|_{\Hf}^2 := \int_{\FC_1}\di y\, |\phi(y)|^2, \]
where $\FC_1$ is a fundamental cell for $\Gamma$ (see \eqref{eqn:defn FCL}). The crystal momentum is effectively defined up to translations in the dual Bravais lattice $\Gamma^*$, consisting of those $\lambda \in \R^d$ such that $\lambda\cdot\gamma\in 2\pi\Z$: indeed,
\[
(\UZ \psi)(k+\gamma^*, y) = \left( \varrho_{\gamma^*} \, \UZ \psi\right)(k , y)   \text{ for all } \gamma^* \in \Gamma^*,
\]
where $(\varrho_{\gamma^*}\varphi)(y):=\E^{-\I \gamma^*\cdot y} \varphi(y)$, and $\varrho \colon \Gamma^* \to \mathcal{U}(\Hf)$, $\gamma^* \to \varrho_{\gamma^*}$, defines a unitary representation. The map defined by \eqref{eqn:defn UZ} extends then to a unitary operator $\UZ \colon \Hi \to \Hi_\varrho$, where $\Hi_\varrho \equiv L^2_{\varrho}(\R^d,\Hf)$ is the space of locally-$L^2$, $\Hf$-valued, $\varrho$-equivariant functions on $\R^d$. Denoting by $\mathbb{B}^d$ a fundamental cell for $\Gamma^*$, the inverse transformation 
$\UZ^{-1} \colon \Hi_\varrho\to\Hi$, is explicitly given by
$$
(\UZ^{-1} \varphi)(x) =\frac{1}{\abs{\mathbb{B}^d}}\int_{\mathbb{B}^d}\di k\,\E^{\iu k\cdot x}\varphi(k,x).
$$

This transform is useful in the analysis of periodic operators as they become \emph{covariant fibered operator} on $\Hi_\varrho$: upon the identification
\[
\Hi_\varrho \equiv L^2_{\varrho}(\R^d,\Hf) \subset L^2(\R^d,\Hf) \simeq \int_{\R^d}^{\oplus} \di k\,\Hf ,
\]
one has
\begin{equation} \label{BFZFiber}
\UZ \, A \, \UZ^{-1} = \int_{\R^d}^{\oplus}\di k\, A(k),
\end{equation}
where each $A(k)$ acts on $\Hf$ and satisfies the covariance property
\[ A(k+\gamma^*) = \varrho_{\gamma^*} \, A(k) \, \varrho_{\gamma^*}^{-1}, \quad \text{for all } k\in\R^d, \; \gamma^*\in\Gamma^*. \]

\subsection{Operator algebras of periodic operators} \label{sec: op alg}

Since the paper relies on the analysis of periodic operators, we will introduce in this Section the necessary operator algebras of operators which have a smooth fiber in the Bloch--Floquet--Zak representation, in an appropriate sense. In the following, $\Hi_1$ and $\Hi_2$ will denote Hilbert subspaces of $\Hf$ (possibly endowed with different norms than the subspace norm) which are left invariant by the action of all momentum-space translation operators $\rho_{\gamma^*}$, $\gamma^* \in \Gamma^*$. As we will specify in the next Sections, in our applications such Hilbert spaces will be either $\Hf$ itself, or the domain $\Df$ of the fiber unperturbed Hamiltonain (endowed with the graph norm of the latter).

\begin{definition}
Let $\mathcal{L}(\Hi_1,\Hi_2)$ denote the space of bounded linear operators from $\Hi_1$ to $\Hi_2$, and $\mathcal{L}(\Hi_1):=\mathcal{L}(\Hi_1,\Hi_1)$. We define
\[
\PO(\Hi_1,\Hi_2) := \left\{\text{ periodic operators $A$ with smooth fibration $\R^d\to \mathcal{L}(\Hi_1,\Hi_2)$, $k\mapsto A(k)$ } \right\}
\]
equipped with the norm 
\[ \norm{A}_{\PO(\Hi_1,\Hi_2)}:=\max_{k\in \mathbb{B}^d}\norm{A(k)}_{\mathcal{L}(\Hi_1,\Hi_2)}. \]
We also set $\PO(\Hi_1) :=\PO(\Hi_1,\Hi_1)$.
\end{definition}

Since the Fr\'echet derivative follows the usual rules of the differential calculus, we have that
\begin{enumerate}
 \item $\PO(\Hi_1,\Hi_2)$ is a linear space;
 \item $\PO(\Hi_1)$ is a normed algebra, as well as $\PO( \Hi_1,\Hi_2 )$ if%
 \footnote{With this inclusion we mean also that the identity $\Hi_2 \hookrightarrow \Hi_1$ is bounded as a map of normed spaces.} %
 $\Hi_2 \subset \Hi_1$;
 \item if $\Hi_2 \subset \Hi_1$, then for $A\in\PO(\Hi_2,\Hi_1)$ and $B\in\PO(\Hi_1,\Hi_2)$ we have
\[
AB\in\PO(\Hi_1)\text{ with }\norm{AB}_{\PO(\Hi_1)}\leq \norm{A}_{\PO(\Hi_2, \Hi_1)}\norm{B}_{\PO(\Hi_1, \Hi_2)}.
\]
\end{enumerate}

It is also useful to consider smooth functions in $\Hi_\varrho \equiv L^2_{\varrho}(\R^d,\Hf)$. As decay at infinity translates into regularity in $k$ via the Bloch--Floquet--Zak transform, for example compactly-supported functions of $x$ are mapped by $\UZ$ to smooth functions of $k$.

\begin{definition} \label{dfn:Cinfty}
We set
\[
C^\infty_\varrho(\R^d,\Hi_1):=\left\{ \varphi\in\Hi_\varrho : \; \varphi(k, \cdot) \in \Hi_1 \text{ for all $k \in \R^d$ and } \varphi\colon\R^d\to\Hi_1\text{ is smooth}\right\}.
\]
\end{definition}

This space of smooth functions, which is clearly dense in $L^2_\varrho(\R^d,\Hi_1)$, is particularly convenient to formulate the invariance of the operator algebras $\PO(\Hi_1,\Hi_2)$ under the derivations given by the commutation with position operators, as detailed in the following statement. Its proof can be found in \cite[Section 3]{MaPaTe}.
\begin{lemma}
\label{lem:derivation of P spaces}
Let $A\in\PO(\Hi_1,\Hi_2)$. Then  
\[
{[A,X_j]}:=\overline{[A,X_j]\Big|_ {\UZ^{-1}\,{C^\infty_\varrho(\R^d,\Hi_1)}}}
\]
is in $\PO(\Hi_1,\Hi_2)$, and
\[
{[A,X_j]}(k)\,\varphi(k)=-\iu\partial_{k_j}A(k)\,\varphi(k)\quad\text{  for all } \varphi \in {C^\infty_\varrho(\R^d,\Hi_1)}.
\]
\end{lemma} 

We conclude this Section by recalling that the space $\B_\infty^\tau$ of bounded periodic operators is endowed with a trace-like functional, called the \emph{trace per unit volume}, defined equivalently as
\[ \tau(A) := \frac{1}{|\FC_1|} \, \Tr_{\Hi}\left(\chi_{\FC_1} \, A \, \chi_{\FC_1}\right) \quad \text{or} \quad \tau(A) := \frac{1}{(2\pi)^d} \int_{\mathbb{B}^d} \di k \, \Tr_{\Hf}(A(k)) \,, \]
whenever the right-hand sides make sense (see Proposition~\ref{prop:tau}). Here, $\chi_{\FC_1}$ is the multiplication operator times the characteristic function of the fundamental cell $\FC_1 \subset \X$. Periodic operators of \emph{trace-per-unit-volume class} define the space $\B_1^\tau$. As we will see shortly, the trace per unit volume is used to compute expectation values of extensive observables in the crystalline, periodic setting which we also employ. We refer the reader to Appendix~\ref{sec:TPUV} and to \cite{MaPaTe} for a list of the relevant properties of the trace per unit volume that will be repeatedly used in the paper.

\subsection{The model}

As stated in the Introduction, our goal is to investigate the response of a crystalline system to the application of an external constant electric field of small intensity. Consequently, a prominent role is played by the Hamiltonian $H_0$ of the system at equilibrium, before the electric field is applied and the response current is probed. Our assumptions on this unpertubed model, which coincide with those adopted in~\cite[Section 3]{MaPaTe}, are stated below.

\begin{assumption} \label{assum:1} We assume the following.  
\begin{enumerate}[label=$(\mathrm{H}_{\arabic*})$,ref=$(\mathrm{H}_{\arabic*})$]
\item \label{item:smooth_H}
The Hamiltonian $H_0$ of the unperturbed system is a self-adjoint periodic operator on $\Hi$, bounded from below. Moreover, its fibers $H_0(k)$, defined in the Bloch--Floquet--Zak representation via~\eqref{BFZFiber}, are self-adjoint operators with a common dense domain $\Df\subset\Hf$. Finally, we assume that $H_0 \in \PO(\Df,\Hf)$, where hereinafter $\Df$ is understood to be equipped with the graph norm $\norm{\,\cdot\,}_{\Df}$ of the operator $H_0(0)$.
\item \label{item:gap} 
We assume the Fermi energy $\mu\in\R$ to lie in a spectral gap of $H_0$. We denote by $\Pi_0 = \chi_{(-\infty, \mu)}(H_0)$ the corresponding spectral projector (Fermi projector). Finally, we assume that\footnote{This assumption is equivalent to require that the fibration $k\mapsto\Pi_0(k)$ takes values in the finite-rank projections on $\Hf$. Indeed, in view of the fact that $\Pi_0$ is an orthogonal projection and the smoothness assumption \ref{item:smooth_H}, it follows that $\mathrm{Rank}(\Pi_0(k))=\Tr(\Pi_0(k)) =m\in\N \cup \set{+ \infty}$ 
 is independent of $k$. Therefore, by virtue of Proposition\ref{prop:tau}\ref{item:per+traceclassfibr} $\Pi_0\in\B_1^\tau$ is equivalent to $m<+ \infty$.} $\Pi_0\in\B_1^\tau$. 
\end{enumerate}
\end{assumption}

The above assumptions are satisfied by a large class of physically relevant models, including tight-binding Hamiltonian of common use in condensed matter physics to model discrete systems, as well as Bloch--Landau operators (under mild regularity assumptions on the electro-magnetic potentials -- see e.g.~\cite[Sec.~3]{MonacoPanatiPisanteTeufel18}) used as continuum models for crystalline systems.

We list below a number of relevant properties which can be deduced from the above Assumption, in combination with Lemma~\ref{lem:derivation of P spaces}, and which will be used repeatedly throughout the paper. As before, we refer the reader to \cite[Sec.~3]{MaPaTe} for a proof.

\begin{proposition}
\label{prop:resolvent-Pi0 der Pi0 and H0}
Under Assumption~\ref{assum:1} we have that
\begin{enumerate}[label=(\roman*), ref=(\roman*)]
\item \label{item:resolvent-Pi0} for every $z\in \rho(H_0)$ the resolvent operator $(H_0 - z \Id)^{-1}$ lies in $\PO( \Hf,\Df )$, and consequently $\Pi_0$ is in $\PO( \Hf,\Df )$ as well;
\item \label{item:der Pi0 and H0} all iterated commutators of $H_0$ with position operators lie in $\PO( \Df,\Hf )$, while all iterated commutators of $\Pi_0$ with position operators lie in $\PO(\Hf,\Df )$. 
\end{enumerate}
\end{proposition}

Having specified the conditions on the model at equilibrium, we drive the system out of equilibrium by introducing an external constant electric field. We choose the direction of this field to be along a preferred coordinate, say $y$. The perturbation will then be modelled space-adiabatically by adding a linear potential to the unperturbed Hamiltonian, namely
\begin{equation}
\label{eqn:pertHam}
H^\eps := H_0 - \eps Y\,,
\end{equation}
where $ \eps \in [0,1]$.

We will be interested in measuring the response to this perturbation of a (possibly different) coordinate of the charge current operator, say along $x$: assuming charge carriers of unit charge,
\begin{equation}
\label{eqn:defn J}
J := \I [H_0, X]. 
\end{equation}
This response will be the current $\tau(J\,\rho)$, where $\rho$ is an appropriate out-of-equilibrium state. The next Section will be devoted to the construction of a projector $\Pi^\eps_n$ which well approximates the state of the system out of equilibrium, at least as far as the response of appropriate observables like the above current is concerned.

\section{Construction of the NEASS to all orders} \label{sec:NEASSAO}

In this Section we generalize the construction of the \emph{non-equilibrium almost-stationary state} (NEASS), realized up to the first order in $\eps$ in \cite[Section 4]{MaPaTe}, to all orders, following the construction performed in the context of interacting models on lattices by \cite{Teufel20} (see \cite{BachmannDeRoeckFraas17, MonacoTeufel17} for related statements in time-dependent adiabatic perturbation theory). For every $n\in\N$ the NEASS, denoted by $\Pi^\eps_n$, is determined uniquely (up to terms of order $\Or(\epsi^{n+1})$) by the following two properties:
\begin{enumerate}[label=(SA$_{\arabic*}$), ref=(SA$_{\arabic*}$)]
\item \label{item:SA1} $\Pi^\eps_n = \E^{-\I \epsi \mathcal{S}^\epsi_n} \, \Pi_0 \, \E^{\I\epsi \mathcal{S}^\epsi_n}$ for some bounded,  periodic and self-adjoint operator $\mathcal{S}^\epsi_n$;
\item \label{item:SA3} $\Pi^\eps_n$ almost-commutes with the Hamiltonian $H^\eps$, namely 
 $[H^\eps,\Pi^\eps_n] = \Or(\eps^{n+1})$.
\end{enumerate}
Here $\Or(\eps^{n+1})$ is understood in the sense of the operator norm. 

\begin{theorem}
\label{thm:constr NEASS}
Consider the Hamiltonian $H^\eps=H_0 - \eps Y$ where $H_0$ satisfies Assumption~\ref{assum:1}. Then, there exists a sequence $\{A_j\}_{j\in\N} \subset \PO(\Hf,\Df)$ such that, setting for any $n\in \N$
\begin{equation}
\label{eqn:defn Sneps}
\SC_n^{\eps}:=\sum_{j=1}^n \eps^{j-1} A_j \, \in \PO(\Hf,\Df)\, ,
\end{equation}
we have that
\begin{equation}
\label{eqn: NEASS prop}
\Pi^\eps_n := \E^{\I \eps \SC_n^{\eps}} \,\Pi_0\, \E^{-\I \eps \SC_n^{\eps}}\qquad \text{ satisfies }\qquad
[H^{\eps} , \Pi^{\eps}_n]=\eps^{n+1}[R^\eps_n,\Pi^{\eps}_n]
\end{equation}
where the map $[0,1]\ni \eps \mapsto R_n^\eps \in \PO(\Hf) \subset \B_{\infty}^{\tau}$ is bounded.
\end{theorem}

\begin{proof}
We start by computing
\begin{align*}
[H^{\eps} , \Pi^{\eps}_n]&=\E^{\I \eps \SC_n^{\eps}}\left[ \E^{-\I \eps \SC_n^{\eps}}H_0\E^{\I \eps \SC_n^{\eps}}-\eps \E^{-\I \eps \SC_n^{\eps}}Y \E^{\I \eps \SC_n^{\eps}}, \Pi_0\right]           \E^{-\I \eps \SC_n^{\eps}}\,.
\end{align*}
Hence, it suffices to choose the operators $A_j$ in such way that there exists $R_n^{\eps}$ uniformly bounded in $\eps$ with
\begin{equation}
\label{eqn:remainder exp}
\left[ \E^{-\I \eps \SC_n^{\eps}}H_0\E^{\I \eps \SC_n^{\eps}}-\eps \E^{-\I \eps \SC_n^{\eps}}Y \E^{\I \eps \SC_n^{\eps}}, \Pi_0\right]=\eps^{n+1}[\E^{-\I \eps \SC_n^{\eps}}R_n^\eps\E^{\I \eps \SC_n^{\eps}},\Pi_0].
\end{equation}

Consider the Taylor expansion in $\lambda$ near $\lambda_0=0$ of the expression
\[ \E^{-\I \lambda \SC_n^{\eps}}B\E^{\I \lambda \SC_n^{\eps}} \,. \]
Evaluating such expansion at $\lambda=\epsi$, one obtains for some $\tilde{\eps}\in [0,\eps]$
\begin{align*}
\E^{-\I \epsi \SC_n^{\eps}}B\E^{\I \epsi \SC_n^{\eps}}&=\sum_{k=0}^n\frac{\eps^k}{k!}\LI_{\SC_n^{\eps}}^k(B)+\frac{\eps^{n+1}}{(n+1)!}\E^{-\I \tilde{\epsi} \SC_n^{\eps}}\LI_{\SC_n^{\eps}}^{n+1}(B)\E^{\I \tilde{\epsi} \SC_n^{\eps}}\cr
&=:\sum_{j=0}^n\eps^j B_{j}+\eps^{n+1}B_{n+1}(\eps).
\end{align*}
On the first line, we use the notation $\LI_A(B):=-\iu[A,B]$, we denoted by $\LI^k_A(B)$ the $k$ nested commutators $[-\iu A,[\dots,[-\iu A,[-\iu A,B]]\dots]]$ for $k\geq 1$ and we set $\LI^0_A(B):=B$.
On the second line, we collected in $B_j$ all the terms of order $\epsi^j$, $0\le j\le n$, coming from the power expansion of $\SC_n^\eps$ as in the statement, while $B_{n+1}(\eps)$ contains contributions from higher-order powers of $\epsi$ and still defines a uniformly bounded function of $\eps$. In particular, each of these coefficients are expressed as nested commutators involving (possibly different) $A_\mu$'s with $B$: for example
\[ B_{0}=B,\quad B_{1}=-\iu[A_1,B],\quad B_{2}=-\frac{1}{2}[A_1, [A_1,B]]-\iu[A_2,B]. \]

We apply the expansion above to $B = H_0$ and a similar expansion, up to order $n-1$, to $B=Y$. We will therefore denote the corresponding coefficients by $H_{0,j}$ and $Y_j$, respectively. Notice that the presence of an extra factor of $\epsi$ in the perturbation $H_\eps - H_0 = -\eps Y$ shifts the indices of the coefficients $Y_j$ by one in the following equations. 

Plugging all the expansions into \eqref{eqn:remainder exp} yields
\begin{equation}
\label{eqn:remainder}
\sum_{j=1}^n\eps^j\left[H_{0,j} - Y_{j-1},\Pi_0 \right]+\eps^{n+1}\left[  H_{0,n+1}(\eps) - Y_{n}(\eps),\Pi_0    \right]=\eps^{n+1}[\E^{-\I \epsi \SC_n^{\eps}}R_n^{\eps}\E^{\I \epsi \SC_n^{\eps}},\Pi_0]
\end{equation}
(notice that the sum on the left-hand side starts from $j=1$ as for $j=0$ we get $[H_{0,0},\Pi_0] = [H_0, \Pi_0] = 0$). Thus, it suffices to determine $A_1,\dots,A_n$ in such a way that for all $j\in\{1,\dots,n\}$
\begin{equation} \label{eqn:recursion}
0=\left[H_{0,j} - Y_{j-1},\Pi_0 \right]\,.
\end{equation}
To this end, it is convenient to notice that
\[ H_{0,j} = \LI_{A_j}(H_0) + L_{j-1} = - \LI_{H_0}(A_j) + L_{j-1} \]
where $L_{j-1}$ involves commutators of $H_0$ with the operators $A_\mu$ with $\mu < j$. Therefore, if we assume that $A_1,\ldots,A_{j-1}$ have been already determined, then $L_{j-1}$ is given and only $A_j$ is still unknown in the above equality: this suggests to determine $A_1, \ldots, A_n$ recursively.

Let us first start then by determining $A_1$ from \eqref{eqn:recursion}. Since $L_0 = 0$ and $Y_0 = Y$, the equation for $A_1$ reads
\[ 0 = [- \LI_{H_0}(A_1)- Y, \Pi_0] \quad \Longrightarrow \quad [ \LI_{H_0}(A_1), \Pi_0] = \LI_{H_0}([A_1,\Pi_0]) = - [Y,\Pi_0]. \]
Notice that the operator $[Y,\Pi_0]$ is off-diagonal with respect to the decomposition $\Hi = \Ran \Pi_0 \oplus (\Ran \Pi_0)^\perp$, i.e.
\[ [Y,\Pi_0] = [Y,\Pi_0]\su{OD} \quad \text{where} \quad T\su{OD} := \Pi_0 \, T \, \Pi_0^\perp + \Pi_0^\perp \, T \, \Pi_0 = \big[[T,\Pi_0],\Pi_0\big] \]
(we denote by $\Pi_0^\perp := \Id-\Pi_0$ the orthogonal projection on $(\Ran \Pi_0)^\perp$). As is well-known (see Appendix \ref{app:invLiou}) the Liouvillian $\LI_{H_0}$ is invertible on such operators, yielding
\[ [A_1,\Pi_0] = \left[A_1\su{OD},\Pi_0\right] = \LI_{H_0}^{-1}(-[Y,\Pi_0]). \]
Taking a further commutator of both sides with $\Pi_0$, we conclude
\[ A_1\su{OD} = \left[\left[A_1\su{OD},\Pi_0\right],\Pi_0\right] = \left[\LI_{H_0}^{-1}(-[Y,\Pi_0]),\Pi_0\right] = -\LI_{H_0}^{-1}(Y\su{OD}). \]
The above considerations hence determine uniquely the off-diagonal part of $A_1$; we may then choose to set
\[ A_1\su{D} := 0. \]

For $1 < j \le n$ we are then required to solve
\begin{multline*}
0= -\LI_{H_0}\(\left[A_j,\Pi_0 \right]\)+\left[L_{j-1},\Pi_0 \right] - \left[Y_{j-1},\Pi_0 \right] \\
 \Longrightarrow \quad \left[A_j,\Pi_0 \right] = \LI_{H_0}^{-1} \left(\left[L_{j-1},\Pi_0 \right] - \left[Y_{j-1},\Pi_0 \right]\right)
 \end{multline*}
or, arguing as above, 
\[ A_j\su{OD} = \LI_{H_0}^{-1} \left(D_{j-1}\right) \quad \text{where} \quad D_{j-1} := \left(L_{j-1} - Y_{j-1}\right)\su{OD}. \]
Observe that $L_{j-1}$ and $Y_{j-1}$ are determined by the previously computed $A_1, \ldots, A_{j-1}$. Once again, we choose $A_j$ to be purely off-diagonal, that is,
\[ A_j\su{D}:=0. \]

In conclusion, we have determined
\[ \SC_n^{\eps}=\sum_{j=1}^n \eps^{j-1} A_j = \LI_{H_0}^{-1} \left( \sum_{\ell=0}^{n-1} \eps^\ell \, D_\ell \right)\in\PO(\Hf,\Df).\]
As each $A_j$ and then $\SC_n^\eps$ are inverse Liouvillians of off-diagonal operators in $\PO(\Hf)$ by virtue of Lemma~\ref{lem:derivation of P spaces}, they are naturally in $\PO(\Hf,\Df)$ as claimed (see Proposition \ref{prop:I(A)}). With this definition of $\SC_n^{\eps}$, it follows by construction that \eqref{eqn:remainder} will be satisfied if we set
\[ R_n^{\eps} := \E^{\I \epsi \SC_n^{\eps}} \left( H_{0,n+1}(\eps) - Y_{n}(\eps)  \right) \E^{-\I \epsi \SC_n^{\eps}}. \]
Clearly, this remainder term is uniformly bounded as a function of $\eps$ with values in $\PO(\Hf)$ in view of the previous discussion on $H_{0,n+1}(\eps)$ and $ Y_{n}(\eps)$.
\end{proof}

\begin{remark}
Clearly the expression $\Pi^\eps_n = \E^{\I \eps \SC_n^{\eps}} \,\Pi_0\, \E^{-\I \eps \SC_n^{\eps}}$ can be used also to obtain a Taylor expansion for the NEASS in powers of $\eps$:
\[ \Pi^\eps_n = \Pi_0 + \eps\,\Pi_1+\eps^2 \,\Pi_2 + \cdots + \eps^n \, \Pi_n + \eps^{n+1} \Pi\sub{reminder}(\eps). \]
The coefficients $\Pi_\ell$ in the above expansion are computable in terms of the $A_j$'s in the statement of Theorem~\ref{thm:constr NEASS}: more specifically, $\Pi_\ell$ will be determined from $A_1, \ldots, A_\ell$. As these $A_j$'s are determined inductively as in the proof (that is, $A_1, \ldots, A_\ell$ determine $A_{\ell+1}$), it is also clear that the above Taylor expansions for the NEASS's $\Pi_n^\eps$ and $\Pi_{n+1}^\eps$ coincide up to order $\eps^n$. Explicit expressions for the coefficients in this Taylor expansion are in any case not needed for the proof of our main Theorem~\ref{thm:main result} in the next Section.
\end{remark}

\noindent We conclude this Section with some immediate consequences from the previous Theorem, which will be used in the next Section.

\begin{corollary}
\label{cor:NEASS}
Consider the Hamiltonian $H^\eps=H_0 - \eps Y$ where $H_0$ satisfies Assumption~\ref{assum:1}. Then we have that for every $n\in\N$
\begin{enumerate}[label=(\roman*), ref=(\roman*)]
\item \label{item:offdiag Heps wrt Pieps} the operator $(\Pi^{\eps}_n)^\perp H^{\eps} \Pi^\eps_n=\eps^{n+1}(\Pi^{\eps}_n)^\perp R_n^{\eps} \Pi^\eps_n$ lies in $\B_{1}^{\tau}$ and the map $[0,1]\ni \eps\mapsto (\Pi^{\eps}_n)^\perp H^{\eps} \Pi^\eps_n\in \B_{1}^{\tau}$ is bounded;
\item \label{item:exp-id} for $\SC_n^{\eps}$ as in \eqref{eqn:defn Sneps}, the operator $\E^{\I \epsi \SC_n^{\eps}}-\Id$ lies in $\PO(\Hf,\Df)$ and the map $[0,1]\ni \eps\mapsto \E^{\I \epsi \SC_n^{\eps}}-\Id \in \PO(\Hf,\Df)$ is bounded;
\item \label{item:reg Pieps} the NEASS operator $\Pi^{\eps}_n$ lies in $\PO(\Hf,\Df)$ and the map $[0,1]\ni \eps\mapsto [X,\Pi^{\eps}_n]\in \PO(\Hf)$ is bounded.
\end{enumerate}
\end{corollary}
\begin{proof}
{\it \ref{item:offdiag Heps wrt Pieps}} The statement is an immediate consequence of \eqref{eqn: NEASS prop}, the fact that $\Pi^{\eps}_n \, (\Pi^{\eps}_n)^\perp = (\Pi^{\eps}_n)^\perp \, \Pi^{\eps}_n = 0$, and the fact that $\Pi^{\eps}_n$ is unitarily equivalent to $\Pi_0$, and is therefore a projection in $\B_1^\tau$ by hypothesis \ref{item:gap}.\\
{\it \ref{item:exp-id}} In view of \eqref{eqn:defn Sneps} $\SC_n^{\eps}\in\PO(\Hf,\Df)$, thus \cite[Lemma 6.4]{MaPaTe} implies the thesis.\\
{\it \ref{item:reg Pieps}} By using Proposition \ref{prop:resolvent-Pi0 der Pi0 and H0}\ref{item:der Pi0 and H0}, \cite[Lemma 6.4]{MaPaTe} and the Leibniz rule, we obtain that $[X,\Pi^{\eps}_n]\in\PO(\Hf)$ and its norm is bounded uniformly in $\eps$.
\end{proof}

\section{Validity of the Kubo formula beyond the linear regime} \label{sec:main result}

We are finally able to state our main result.

\begin{theorem}
\label{thm:main result}
Consider the Hamiltonian $H^\eps=H_0 - \eps Y$ where $H_0$ satisfies Assumption~\ref{assum:1}. Then for every $ n\in\N$ we have that
\[
\tau(J \,\Pi^\eps_n)=\eps\,\sigma\sub{Hall}+\Or(\eps^{n+1}),
\]
where $J$ is the charge current operator in~\eqref{eqn:defn J}, the NEASS $\Pi^\eps_n$ is as in the statement of Theorem~\ref{thm:constr NEASS}, and
\[
\sigma\sub{Hall}:=\iu \tau(\Pi_0\left[[\Pi_0,X],[\Pi_0,Y]\right]).
\]
\end{theorem}

The above Theorem states that the conductivity associated to the response of the current operator $J$ is given by the Hall conductivity $\sigma\sub{Hall}$, a quantity which is defined only through the equilibrium Fermi projection $\Pi_0$ and which emerges at the linear level (Kubo formula), up to orders which are arbitrarily high in the strength of the perturbing electric field. The result thus establishes the validity of the Kubo formula for this conductivity also beyond linear response.

The proof of the above Theorem relies on a number of intermediate steps, which we detail first.

\subsection{A useful lemma}

As a first tool to be used in the argument for the main Theorem~\ref{thm:main result}, we prove the following

\begin{proposition} \label{prop:useful}
Let $P$ be a projection on $\Hi$ such that $P \in \PO(\Hf) \cap \B_{1}^{\tau}$. Assume that the operator $A$ is such that $PAP \in \PO(\Hf)$. Then for all $j\in\set{1,\ldots,d}$ the trace per unit volume of the commutator $[PAP, PX_jP]$ is well-defined and 
\[ \tau([PAP, PX_jP])=0. \]
\end{proposition}
\begin{proof}
We observe that
\[ [PAP, PX_jP] = [PAP, X_j] - [PAP, X_j\su{OD}] \]
where $X_j\su{OD}$ refers to the off-diagonal decomposition of the operator $X_j$ with respect to the projection $P$, and the equality is first established on the dense subspace $\UZ^{-1}C^\infty_\varrho(\R^d,\Hf)$. Proposition~\ref{prop:tau}\ref{item:tau[A,X]=0} implies that the first summand on the right-hand side has vanishing trace per unit volume, since $PAP \in \PO(\Hf)\cap \B_1^\tau$. On the other hand, the second summand is $\B_1^\tau$ in view of the hypothesis that $P \in \PO(\Hf)\cap \B_1^\tau$ (see Lemma \ref{lem:derivation of P spaces}); invoking Proposition ~\ref{prop:resolvent-Pi0 der Pi0 and H0}\ref{item:der Pi0 and H0} and Lemma~\ref{prop:cycl of tau}, we conclude that its trace per unit volume vanishes.
\end{proof}

\begin{remark} \label{rmk:ChernMarker}
Observe that the above Proposition does \emph{not} apply to the operator $A=X_i$, $i\ne j$; indeed (compare e.g. \cite[Eq. (2.15)]{MarcelliMoscolariPanati})
\[ [PX_iP, PX_jP] = P\big[[P,X_i],[P,X_j]\big]P, \]
where the equality is first established on $\UZ^{-1}C^\infty_\varrho(\R^d,\Hf)$. Therefore the trace per unit volume $\tau([PX_iP, PX_jP])$ equals (up to a factor $2\pi\iu$) the \emph{Chern marker} of the projection $P$, which may very well be non-zero. 

As we saw in the statement of the main result, the Chern marker of the Fermi projection $\Pi_0$ defines the linear response coefficient $\sigma\sub{Hall}$. It's worth mentioning that, in the present context of periodic operators, the (non-)vanishing of the Chern marker has been linked to the (non-)existence of localized orthogonal \emph{Wannier functions} spanning the Fermi projection $\Pi_0$ \cite{MonacoPanatiPisanteTeufel18, CorneanMonacoMoscolari18, Gianfa}. In a more general setting, where periodicity is broken, one can still define a generalized notion of Wannier basis for an isolated spectral island, but the relation of its existence with the Chern marker remains to be fully understood. In fact, while the existence of a generalized Wannier basis (with suitable localization) has been shown to imply the vanishing of the Chern marker (see \cite{MarcelliMoscolariPanati, LuStubbs}), the converse implication \cite{LuStubbsWatson} remains a challenging and interesting line of research.
\end{remark}

\subsection{Chern--Simons formula}
Following the previous Remark, we prove an analogue of the \emph{Chern--Simons formula} which establishes the invariance of the Chern marker of a projection under unitary conjugation. This formula, well rooted in differential geometry and bundle theory, was exploited in \cite{KleinSeiler} in a context similar to ours. 

\begin{proposition}[Chern--Simons formula]
\label{prop:CS formula}
Let $P \in \PO(\Hf) \cap \B_{1}^{\tau}$ be a projection and $U \in \PO(\Hf)$ be unitary. Define $P_U := U P U^{-1}$. Then
\[ \tau([P_UX_iP_U, P_UX_jP_U]) = \tau([PX_iP, PX_jP]). \]
\end{proposition}
\begin{proof}
Write
\begin{align*}
U^{-1}[P_UX_iP_U&, P_UX_jP_U]U  = [PU^{-1}X_iUP, PU^{-1}X_jUP] \\
& = [PX_iP, PX_jP]  + \big[PU^{-1}[X_i,U]P, PX_jP\big] \\
& \quad+ \big[PX_iP, PU^{-1}[X_j,U]P\big] + \big[PU^{-1}[X_i,U]P, PU^{-1}[X_j,U]P\big].
\end{align*}
Notice now that, with the standing assumptions, the operators $U^{-1}[X_i,U]$ and $U^{-1}[X_j,U]$ are in $\PO(\Hf)$ by virtue of Lemma \ref{lem:derivation of P spaces}. Using Proposition \ref{prop:useful} and Lemma \ref{prop:cycl of tau}, the conclusion follows.
\end{proof}

\subsection{Proof of Theorem~\ref{thm:main result}}
First of all, notice that by the very definition of the current operator $J$ in \eqref{eqn:defn J} and the construction of the NEASS $\Pi^\eps_n$ in Theorem~\ref{thm:constr NEASS} 
\begin{align*}
J \,\Pi^\eps_n&=\iu [H_0,X] \(\E^{\I \eps \SC_n^{\eps}}-\Id \)\Pi_0\, \E^{-\I \eps \SC_n^{\eps}}+\iu [H_0,X]  \Pi_0\, \E^{-\I \eps \SC_n^{\eps}},
\end{align*}
where each summand on the right-hand side is in $\B_{1}^{\tau}$. Indeed, for the first summand observe that
\[
[H_0,X]\cdot \(\E^{\I \eps \SC_n^{\eps}}-\Id \)\cdot \Pi_0\cdot \E^{-\I \eps \SC_n^{\eps}}\in\PO(\Df,\Hf)\cdot \PO(\Hf,\Df)\cdot \B_{1}^{\tau}\cdot \PO(\Hf)\subset \PO(\Hf)\cap\B_{1}^{\tau},
\]  
by applying Proposition~\ref{prop:resolvent-Pi0 der Pi0 and H0}\ref{item:der Pi0 and H0}, Corollary~\ref{cor:NEASS}\ref{item:exp-id}, and hypothesis \ref{item:gap}. Similarly, for the second summand note that
\[
[H_0,X] \cdot \Pi_0 \cdot \Pi_0\, \E^{-\I \eps \SC_n^{\eps}}\in\PO(\Df,\Hf)\cdot \PO(\Hf,\Df)\cdot \B_{1}^{\tau}\cdot \PO(\Hf)\subset \PO(\Hf)\cap\B_{1}^{\tau}.
\]
In view of the cyclicity of the trace per unit volume, Corollary~\ref{cor:NEASS}\ref{item:offdiag Heps wrt Pieps} and Corollary~\ref{cor:NEASS}\ref{item:reg Pieps}, we have that\footnote{We are allowed to perform all the following algebraic manipulations since the range of $\UZ \chi_{\FC_1}$ is contained in $C^\infty_\varrho(\R^d,\Hf)$ (see Definitions \ref{dfn:Cinfty} and \ref{defn:tuv}).}
\begin{equation} \label{eqn:JPiesp}
\begin{aligned}
&\tau\left( [H_0,X]\Pi^\eps_n \right)=\tau\left( \Pi^\eps_n [H^\eps,X]\Pi^\eps_n \right)\cr
&=\tau\left(  [\Pi^\eps_n H^\eps \Pi^\eps_n ,\Pi^\eps_n X \Pi^\eps_n] \right)+\eps^{n+1}\,\tau\left( \Pi^\eps_n R_n^{\eps}\left(\Pi^{\eps}_n\right)^\perp X \Pi^\eps_n - \Pi^\eps_n X \left(\Pi^{\eps}_n\right)^\perp R_n^{\eps} \Pi^\eps_n\right)\cr
&=\tau\left(  [\Pi^\eps_n H^\eps \Pi^\eps_n ,\Pi^\eps_n X \Pi^\eps_n] \right)+\eps^{n+1}\,\tau\left( \Pi^\eps_n \big[[\Pi^\eps_n, R_n^{\eps}], [X ,\Pi^\eps_n ]\big] \Pi^\eps_n\right)\cr
&=\tau\left(  [\Pi^\eps_n H_0 \Pi^\eps_n ,\Pi^\eps_n X \Pi^\eps_n] \right)-\eps\,\tau\left(  [\Pi^\eps_n Y \Pi^\eps_n ,\Pi^\eps_n X \Pi^\eps_n] \right)\\
&\phantom{=}+\eps^{n+1}\,\tau\left( \Pi^\eps_n \big[[\Pi^\eps_n, R_n^{\eps}], [X ,\Pi^\eps_n ]\big] \Pi^\eps_n\right)\cr
\end{aligned}
\end{equation}
where the term carrying the prefactor $\eps^{n+1}$ is uniformly bounded in $\eps$. Observe that in view of Corollary~\ref{cor:NEASS}\ref{item:exp-id} and Proposition~\ref{prop:resolvent-Pi0 der Pi0 and H0}\ref{item:resolvent-Pi0} we have that
\[
\Pi^\eps_n H_0 \Pi^\eps_n=\Pi^\eps_n H_0\cdot(\E^{\I \epsi \SC_n^{\eps}}-\Id)\Pi_0\E^{-\I \epsi \SC_n^{\eps}}+\Pi^\eps_n H_0\cdot\Pi_0\E^{-\I \epsi \SC_n^{\eps}}\in\PO(\Hf,\Df)\cdot\PO(\Df,\Hf)\subset \PO(\Hf),
\]
thus Proposition~\ref{prop:useful} implies that the first summand on the right-hand side of~\eqref{eqn:JPiesp} vanishes. On the other hand, by Corollary~\ref{cor:NEASS}\ref{item:exp-id} the unitary $\E^{\I \epsi \SC_n^{\eps}}$ is in $\PO(\Hf)$, therefore Proposition~\ref{prop:CS formula} and Remark~\ref{rmk:ChernMarker} imply that the second summand in \eqref{eqn:JPiesp} can be rewritten as
\[
-\eps\,\tau(  [\Pi^\eps_n Y \Pi^\eps_n ,\Pi^\eps_n X \Pi^\eps_n] )=\eps\,\tau(  [\Pi_0 X \Pi_0 ,\Pi_0 Y \Pi_0] )=\eps\,\tau(\Pi_0\big[[\Pi_0,X],[\Pi_0,Y]\big]\Pi_0).
\]
This concludes the proof. \qed


\appendix
\section{Trace per unit volume} \label{sec:TPUV}

Here we recall the definition and the main properties of the trace-per-unit-volume functional (for further details see \cite[Section 2]{MaPaTe} and references therein). For any $L\in 2\N+1$, we define
\begin{equation}
\label{eqn:defn FCL}
\FC_L:=\left\lbrace x\in \X : x = \sum_{j=1}^d \alpha_j \, a_j \text{ with }  |\alpha_j|\leq L/2 \; \forall\: j \in \set{1,\ldots,d}\right\rbrace
\end{equation}
and $\chi_L:=\chi_{\FC_L}$, denoting the orthogonal projection on $\Hi$ which multiplies by the characteristic function of $\FC_L$. In particular, the set $\FC_1$ is called a \emph{fundamental cell}. 

\noindent We say that an operator $A$ acting in $\Hi$ is \emph{trace class on compact sets} if and only if $\chi_K A \chi_K$ is trace class for all compact sets $K \subset \X$
\footnote{Notice that in the discrete case this condition is automatically satisfied  for any operator $A$ because the range of $\chi_K$ is finite-dimensional.}. 

\begin{definition}[Trace per unit volume]
\label{defn:tuv}
Let $A$ be an operator acting in $\Hi$ such that $ A $ is trace class on compact sets. The \emph{trace per unit volume} of $A$ is defined as
\begin{equation}
\label{eqn:defn tau}
\tau(A):=\lim_{\substack{L\to\infty\\L\in 2\N+1}}\frac{1}{\abs{\FC_L}}\Tr(\chi_L A \chi_L),
\end{equation}
whenever the limit exists.
\end{definition}
\noindent The most relevant properties of the trace per unit volume are presented in the following results, whose proofs can be found in~\cite[Section 2]{MaPaTe}. We detail an argument only for Proposition~\ref{prop:tau}\ref{item:tau[A,X]=0}, which is not present in the above reference.

\noindent We introduce the vector spaces
\begin{gather*}
\B_\infty^\tau := \set{\text{bounded \emph{periodic} operators on } \Hi}, \\
\B_1^\tau := \set{\text{$A\in \B_\infty^\tau$ such that } \norm{A}_{1,\tau}:=\tau(\abs{A}) < \infty}.
\end{gather*}
We recall that $\B_1^\tau$ is invariant by left and right multiplication by elements of $\B_\infty^\tau\supset \PO(\Hf)$. Similarly to the standard trace, the trace per unit volume is (conditionally) \emph{cyclic}.
\begin{lemma}[Cyclicity of the trace per unit volume]
\label{prop:cycl of tau}
If $A\in\B_1^\tau$ and $B\in\B_\infty^\tau$, then $\tau(AB)=\tau(BA).$
\end{lemma}

\noindent The next result collects all the essential properties of the trace per unit volume.
\begin{proposition} 
\label{prop:tau}
\begin{enumerate}[label=(\roman*), ref=(\roman*)]
\item \label{item:taucont} Let $A\in \B_1^\tau$. Then 
\[
\Tr(\abs{\chi_L A \chi_L})<\infty\quad\forall\, L\in 2\N+1.
\]
In particular, we have that $A$ is trace class on compact sets.
\item \label{item:per+traceclassoncompact} Let $A$ be periodic and trace class on compact sets. Then $\tau(A)$ is well-defined and 
\[
\tau(A) =\frac{1}{\abs{\FC_1}} \Tr(\chi_1 A \chi_1).
\]
\item \label{item:per+traceclassfibr} Let $A$ be a periodic and bounded operator acting on $\Hi$.
Denoting by
\[ \UZ \, A \, \UZ^{-1} = \int_{\mathbb{R}^d}^{\oplus}\di k\, A(k)   \]
its Bloch--Floquet--Zak decomposition, assume that $A(k)$ is trace class and that $\Tr_{\Hf}(\abs{A(k)})<C$ for all $k\in \mathbb{B}^d$. Then
\[
\Tr(\chi_1 A \chi_1) = \frac{1}{|\mathbb{B}^d|} \int_{\mathbb{B}^d}  \di k\,\Tr_{\Hf}(A(k)). 
\]
\item \label{item:tau[A,X]=0} Let $A\in \PO(\Hf)\cap\B_1^\tau$. Then $\tau([A,X_i])$ is well-defined and $\tau([A,X_i])=0$ for every $1\leq i\leq d$.
\end{enumerate} 
\end{proposition}
\begin{proof}
{\it \ref{item:tau[A,X]=0}} By Lemma~\ref{lem:derivation of P spaces} we have that $[A,X_i]\in\PO(\Hf)$. Observe that the operator
\[
\chi_L [A,X_i] \chi_L=\chi_L A \chi_L X_i \chi_L-\chi_L X_i \chi_L A \chi_L\, \text{ is trace class on compact sets},
\]
by applying Proposition~\ref{prop:tau}\ref{item:taucont} and noticing that $\chi_L X_i \chi_L$ is bounded. Thus, Proposition~\ref{prop:tau}\ref{item:per+traceclassoncompact} implies that
$$
\abs{\FC_1}\cdot \tau([A,X_i])=\Tr(\chi_1 [A,X_i] \chi_1)=\Tr(\chi_1 A\chi_1 X_i \chi_1-\chi_1 X_i\chi_1 A \chi_1),
$$
where both summands inside the trace are trace class because $\chi_1 A \chi_1$ is trace class. The cyclicity of the standard trace concludes the proof.
\end{proof}

\section{Inverse Liouvillian}
\label{app:invLiou}

Here we recall the expression of the inverse Liouvillian $\LI_{H_0}^{-1}$, associated with the unperturbed Hamiltonian $H_0$, and its relevant properties. 

We look for the solution $B$ to the equation $\LI_{H_0}(B)=-\iu [H_0,B] = A$, where $A=A\su{OD}\in \PO(\Hf)$  is off-diagonal with respect to the decomposition $\Hi = \Ran \Pi_0 \oplus (\Ran \Pi_0)^\perp$. We state in the following Proposition, whose proof can be found in \cite[Subsection 6.2]{MaPaTe}, the solution to this problem, which traces back at least to \cite[Equation (2.11)]{AvronSeilerYaffe} (see also \cite[Equation (A10)]{KleinSeiler}).

\begin{proposition} \label{prop:I(A)}
Under Assumption~\ref{assum:1}, let $A\in \PO(\Hf)$ be such that $A=A\su{OD}$ with respect to $\Pi_0$. Then the unique off-diagonal solution in $\PO(\Hf,\Df)$ to the equation
\[
\LI_{H_0}(B)=-\iu [H_0,B] = A \qquad  \text{on} \quad \UZ^{-1}L^2_{\varrho}(\R^d,\Df),
\]
is given by
\begin{equation}
\label{eqn:invLiou}
B=\LI_{H_0}^{-1}(A):= \frac{1}{2 \pi} \oint_C \di z \, (H_0 - z \Id)^{-1} \, [\Pi_0,A] \, (H_0 - z \Id)^{-1}.
\end{equation}
\end{proposition}


\bigskip \bigskip

{\footnotesize

\begin{tabular}{ll}
(G. Marcelli)
             &  \textsc{Mathematics Area, SISSA} \\ 
        	&   Via Bonomea 265, 34136 Trieste, Italy \\
        	&  {E-mail address}: \href{mailto:giovanna.marcelli@sissa.it}{\texttt{giovanna.marcelli@sissa.it}}\\[10pt]
(D.~Monaco) & \textsc{Dipartimento di Matematica, ``La Sapienza'' Universit\`{a} di Roma} \\
 &  Piazzale Aldo Moro 2, 00185 Rome, Italy \\
 &  {E-mail address}: \href{mailto:monaco@mat.uniroma1.it}{\texttt{monaco@mat.uniroma1.it}} \\
\end{tabular}
}

\end{document}